\documentclass[a4paper,USenglish,cleveref, autoref, thm-restate,pdfa]{lipics-v2021}
\nolinenumbers
\pdfoutput=1

\usepackage{enumerate}
\usepackage{amsfonts}
\usepackage{booktabs}
\usepackage{siunitx}
\usepackage{multirow}
\usepackage{graphicx}
\usepackage{listings}
\usepackage{commath}
\usepackage{epstopdf}
\usepackage{url}
\usepackage{mathtools}
\usepackage{alltt}
\usepackage{mathrsfs}
\usepackage{wrapfig}

\usepackage{subcaption}
\usepackage[normalem]{ulem}
\usepackage{etoolbox}
\usepackage{stmaryrd}
\usepackage{amsmath,amssymb}
\usepackage[noend]{algorithmic}

\usepackage{mathtools}
\usepackage{etoolbox}
\usepackage{dsfont}
\usepackage{color}
\usepackage{colortbl}
\usepackage{multirow}
\usepackage[scaled]{helvet}

\usepackage{fullminipage}

\usepackage{graphicx,fancyvrb}
\usepackage[colorinlistoftodos]{todonotes}
\usepackage[noend]{algorithmic}
\usepackage{enumitem}
\usepackage{booktabs}
\usepackage{capt-of}
\usepackage{xcolor}
\usepackage{multirow}

\usepackage{float}
\usepackage{subfloat}

\usepackage{array}
\usepackage{arydshln}
\usepackage{todonotes}

\usepackage{amsthm}

\theoremstyle{definition}

\definecolor{mygreen}{rgb}{0,0.6,0}
\definecolor{mygray}{rgb}{0.5,0.5,0.5}
\definecolor{mymauve}{rgb}{0.58,0,0.82}

\usepackage{listings}
\lstnewenvironment{VerbatimText}[1][]{
    
    \lstset{fancyvrb=true,basicstyle=\footnotesize,captionpos=b,xleftmargin=2em,#1}
}{}

\newcommand{\ignore}[1]{}

\newcommand{\batya}[1]{{\texttt{\color{blue} Batya: [{#1}]}}}

\newcommand{\A}{\texttt{A}}

\newcommand*{\rom}[1]{\expandafter\@slowromancap\romannumeral #1@}
\newcommand{\RNum}[1]{\uppercase\expandafter{\romannumeral #1\relax}}

\newtheorem{defn}{Definition}[section]

\newcommand{\eat}[1]{}

\eat{

\newtheorem{corollary}[defn]{Corollary}

\newtheorem{proposition}[defn]{Proposition}

\newcommand{\sel}[1]{{\sigma}}

}

\newcommand{\cut}[1]{}

\newcommand{\triple}[3]{\langle{#1},{#2},{#3}\rangle}

\def\set#1{\mathord{\{#1\}}}

\def\eqdef{\mathrel{\stackrel{\textsf{\tiny def}}{=}}}

\def\A{\mathcal{A}}

\def\e#1{\emph{#1}}
\newenvironment{citedtheorem}[1]
{\begin{theorem}{\it\e{(#1)}}\,\,}
	{\end{theorem}}
\newenvironment{citedlemma}[1]
{\begin{lemma}{\it\e{(#1)}}\,\,}
	{\end{lemma}}
\newenvironment{citeddefinition}[1]
{\begin{definition}{\it\e{(#1)}}\,\,}
	{\end{definition}}

\newenvironment{repeatresult}[2]
{\vskip0.5em\par\textsc{#1} #2.\em}
{\vskip1em}

\def\appendix{\par
	\section*{APPENDIX}
	\setcounter{section}{0}
	\setcounter{subsection}{0}
	\def\thesection{\Alph{section}} }

\def\eqdef{\mathrel{\stackrel{\textsf{\tiny def}}{=}}}

\def\B{\mathcal{B}}
\def\e#1{\emph{#1}}

\newcommand{\algname}[1]{{\sf #1}}
\def\myrulewidth{2.80in}
\def\therule{\rule{\myrulewidth}{0.2pt}}

\def\myrulewidthwide{4in}
\def\therulewide{\rule{\myrulewidthwide}{0.2pt}}

\newenvironment{algserieswide}[2]
{\centering\begin{figure}[#1]\begin{center}\def\thecaption{\caption{#2}}
			\begin{tabular}{p{\myrulewidthwide}}\therulewide\end{tabular}\vskip0.2em}
		{\thecaption\end{center}\end{figure}}

\newenvironment{insidealgwide}[2]
{\normalsize\begin{insidecodewide}{#1}{#2}{Algorithm}}
	{\end{insidecodewide}}

\newenvironment{insidecode}[3]
{
	\begin{tabular}{p{\myrulewidth}}
		\multicolumn{1}{c}{\rule{0mm}{3mm}{\bf #3} $\algname{#1}(\mbox{#2})$\vspace{-0.6em}}\\
		\therule\vskip-0.8em\therule
		\vspace{-1em}
		\begin{algorithmic}[1]}
		{\end{algorithmic}
		\vskip-0.4em\therule
\end{tabular}}

\newenvironment{insidecodewide}[3]
{
	\begin{tabular}{p{\myrulewidthwide}}
		\multicolumn{1}{c}{\rule{0mm}{3mm}{\bf #3} $\algname{#1}(\mbox{#2})$\vspace{-0.6em}}\\
		\therulewide\vskip-0.8em\therulewide
		\vspace{-1em}
		\begin{algorithmic}[1]}
		{\end{algorithmic}
		\vskip-0.3em\therulewide
\end{tabular}}

\newcommand{\T}{{\mathcal{T}}}

\newcommand{\comp}[1]{\overline{#1}}

\newcommand{\minlsep}[2]{\mathcal{S}_{#1}(#2)}
\newcommand{\minlsepG}[3]{\mathcal{S}^{#1}_{#2}(#3)}
\newcommand{\minlsepst}[1]{\mathcal{S}_{s,t}(#1)}
\newcommand{\minsepst}[1]{\mathcal{L}_{s,t}(#1)}
\newcommand{\impsepst}[1]{\mathcal{S}_{s,t}^*(#1)}
\newcommand{\impsepstk}[1]{\mathcal{S}_{s,t,k}^*(#1)}

\newcommand{\minlsepEst}[2]{\mathcal{S}_{s,t}(#1,\comp{#2})}
\def\minsep{\mathcal{L}}

\def\sat{\mathrm{Sat}}

\newcommand{\edges}{\texttt{E}}

\newcommand{\nodes}{\texttt{V}}

\definecolor{mygreen}{rgb}{0,0.6,0}
\definecolor{mygray}{rgb}{0.5,0.5,0.5}
\definecolor{mymauve}{rgb}{0.58,0,0.82}
\definecolor{cadmiumgreen}{rgb}{0.0, 0.42, 0.24}

\def\gq1{{\geq}1}

\def\cc{\mathcal{C}}

\newif \ifnonplanar
\nonplanarfalse 

\newcommand{\sminus}{\scalebox{0.75}[1.0]{$-$}}

\usepackage{setspace}\setdisplayskipstretch{}

\usepackage{wrapfig}
\title{Listing Small Minimal $s,t$-separators in FPT-Delay} 
\titlerunning{}

\hideLIPIcs
 
\eat{
\author{ }{ }{}{}{}
}

\author{Batya Kenig}{Technion, Israel Institute of Technology}{batyak@technion.ac.il}{}{}

\eat{
\author{John Q. Public}{Dummy University Computing Laboratory, [optional: Address], Country \and My second affiliation, Country \and \url{http://www.myhomepage.edu} }{johnqpublic@dummyuni.org}{https://orcid.org/0000-0002-1825-0097}{(Optional) author-specific funding acknowledgements}

}

\authorrunning{Batya Kenig} 

\ccsdesc[100]{ }
\keywords{minimal separators, ranked enumeration} 

\begin{document}

\maketitle

\begin{abstract}
	Let $G$ be an undirected graph, and $s,t$ distinguished vertices of $G$. A minimal $s,t$-separator is an inclusion-wise minimal vertex-set whose removal places $s$ and $t$ in distinct connected components. 
	We present an algorithm for listing the minimal $s,t$-separators of a graph, whose cardinality is at most $k$, with FPT-delay, where the parameter depends only on $k$. This problem finds applications in various algorithms parameterized by treewidth, which include query evaluation in relational databases, probabilistic inference, and many more. We also present a simple algorithm that enumerates all of the (not necessarily minimal) $s,t$-separators of a graph in ranked order by size.
\end{abstract}

\section{Introduction}
\label{sec:introduction}
Let  $G(V,E)$ be a finite, undirected graph with nodes $V=\nodes(G)$ and edges $E=\edges(G)$.
For two distinguished vertices $s,t\in \nodes(G)$, an $s,t$-vertex-separator is a subset $S\subseteq \nodes(G)$ such that $s$ and $t$ reside in distinct connected components of the graph that results from $G$ by removing vertices $S$ and their adjacent edges; an $s,t$-edge-separator is a subset $A\subseteq \edges(G)$ such that $s$ and $t$ reside in distinct connected components of the graph $G(V,E\setminus A)$.
Vertex and edge separators are widely studied in graph theory, with a long and deep history that is motivated by a wide range of applications in computer science~\cite{DBLP:books/cu/NI2008}.
Finding optimal (e.g., minimum weight) edge and vertex separators are fundamental problems, that are solved in polynomial time using network flow techniques~\cite{DBLP:journals/jal/HenzingerRG00,DBLP:journals/siamcomp/EvenT75,DBLP:books/daglib/0032640}. A natural extension is the task of their enumeration, and in particular, enumerating edge and vertex-separators in ranked order by their size or weight.\eat{A polynomial-delay algorithm for enumerating $s,t$-edge-separators in ascending order of their weight was presented by Vazirani and Yannakakis~\cite{DBLP:conf/icalp/VaziraniY92}, and before that, an algorithm for enumerating the best $K$ $s,t$-edge-separators was presented by Hamacher~\cite{DBLP:journals/orl/Hamacher82}. Both of these algorithms apply the known partitioning technique due to Lawler~\cite{Lawler1972,DBLP:journals/siamcomp/LawlerLK80,DBLP:journals/pvldb/GolenbergKS11}.} In this work, we focus on $s,t$-vertex-separators~\cite{DBLP:journals/ijfcs/BerryBC00,DBLP:journals/siamcomp/KloksK98,DBLP:journals/dam/Takata10}.
An $s,t$-separator $S\subseteq \nodes(G)$ is a \e{minimal $s,t$-separator} if no strict subset of $S$ is an $s,t$-separator. We say that $S$ is a \e{minimum $s,t$-separator} if $|S'|\geq |S|$ for every $s,t$-separator $S'$. 

Enumerating minimal separators of bounded cardinality refines and extends two well-studied enumeration problems~\cite{DBLP:journals/corr/abs-2012-09153}: enumeration of all minimal separators, and enumeration of all minimum separators~\cite{DBLP:conf/soda/Kanevsky90}.
Berry et al.~\cite{DBLP:journals/ijfcs/BerryBC00} developed an efficient algorithm that lists the minimal separators of a graph with a delay of $O(|\nodes(G)|^3)$ between consecutive outputs. The algorithm of Berry et al.~\cite{DBLP:journals/ijfcs/BerryBC00}, as well as others~\cite{DBLP:journals/siamcomp/KloksK98,DBLP:journals/dam/Takata10,SHEN1997169}, does not list the minimal separators in any ranked order. Kanevsky~\cite{DBLP:conf/soda/Kanevsky90} developed a complicated algorithm that enumerates all the minimum separators of a graph. In this work, we present an algorithm that given an undirected graph $G$ with $n$ vertices and $m$ edges, and a  bound $k$, lists all of the minimal $s,t$-separators whose size is at most $k$ with delay $O(p(n,m)\cdot 4^k)$ where $p$ is a polynomial in $n,m$.   Recently, Korhonen presented an algorithm that lists the minimal $s,t$-separators whose size is at most $k$ in incremental-polynomial-time~\cite{DBLP:journals/corr/abs-2012-09153}. We improve this to FPT-delay. We also present a simple algorithm that enumerates the $s,t$-separators in ranked order by size.
\eat{ Finally, a strict subset of the set of all minimal $st$-separators known as \e{important separators}~\cite{DBLP:books/sp/CyganFKLMPPS15,DBLP:journals/tcs/Marx06} posses a special structure that allows enumerating all important $st$-separators of size at most $q$, in polynomial delay~\cite{DBLP:journals/jacm/ChenLLOR08}.}

{\bf Motivation and Applications.} Fixed Parameter Tractable (FPT) algorithms~\cite{DBLP:books/sp/CyganFKLMPPS15,DBLP:series/txtcs/FlumG06} are those with running time $O(n^cf(k))$, for a constant $c$ independent of both the size of the input $n$, and the parameter $k$, and $f$ is a
computable function. The \e{treewidth} of a graph is an important graph complexity measure that
occurs as a parameter in many FPT algorithms~\cite{DBLP:series/txtcs/FlumG06,DBLP:journals/cj/BodlaenderK08} that operate over the \e{tree decomposition} of the graph. 
A tree decomposition of a graph $G$ is a tree $\T$ such that each
node of $\T$ is associated with a bag of vertices of $G$, every
edge in $\edges(G)$ is contained in at least one bag, and every vertex
in $\nodes(G)$ spans a connected subtree of $\T$. 
The \e{width} of a tree decomposition is defined as the size of the largest bag minus one. The \e{treewidth} of a graph $G$ is the minimum width among all possible tree decompositions of $G$. 

A tree decomposition is \e{proper} if it cannot be ``improved'' by removing or splitting a bag~\cite{DBLP:journals/dam/CarmeliKKK21}. There is a one-to-one correspondence between the \e{minimal triangulations} and the proper tree decompositions of a graph $G$~\cite{DBLP:journals/dam/CarmeliKKK21}. 
A graph $G$ is \e{chordal} if every cycle of $G$ with four or more vertices contains a chord.
A \e{triangulation} of $G$ is a graph $G'$ that is obtained from $G$ by adding edges so that $G'$ is chordal. The cardinality of the set of added edges is called the \e{fill-in} of the triangulation $G'$. A triangulation is \e{minimal} if no triangulation can be obtained using a strict subset of the added edges. It is shown in~\cite{DBLP:journals/dam/CarmeliKKK21} that every edge in a proper tree decomposition represents a minimal separator in $G$; the intersection of the two bags adjacent to the edge is a minimal separator of $G$. Therefore, a necessary condition for a (proper) tree decomposition to have low width, is that the minimal separators represented by its edges have bounded cardinality. Specifically, a proper tree decomposition whose width is $k$ can only represent minimal separators (i.e., via its edges) whose size is at most $k$~\cite{DBLP:conf/sea2/Tamaki19}.

\eat{
Every edge of the tree decomposition represents a minimal separator in $G$; the conjunction of the two bags adjacent to the edge are a minimal separator of $G$~\cite{DBLP:journals/dam/CarmeliKKK21}.
The runtime of algorithms that operate over a tree decomposition is exponential in the \e{width} of the tree decomposition, which is the cardinality of the largest bag (minus one). Another common measure for the quality of a tree decomposition is \e{fill in}, the number of missing edges among bag neighbors. Consequently, a necessary condition for a tree decomposition to have low width (or fill-in) is that it represents minimal separators of small cardinality (small number of ``fill'' edges).
Tree decompositions, and their generalizations to hypergraphs-- \e{hypertree decompositions}, are applied in a plethora of applications in databases and beyond. These include optimization of join queries in databases~\cite{DBLP:journals/jcss/ScarcelloGL07,10.1145/846241.846269,10.1145/2723372.2764946,DBLP:conf/sebd/GottlobLS99,DBLP:journals/sigmod/NgoRR13}, RNA analysis in bioinformatics~\cite{DBLP:conf/wabi/ZhaoMC06}, and inference in probabilistic graphical models~\cite{koller2009probabilistic}.
}

Indeed, while computing an optimal tree decomposition, with the smallest width, is NP-hard~\cite{doi:10.1137/0608024}, current state-of-the-art algorithms for finding an optimal tree-decomposition require the set of minimal separators with bounded cost~\cite{DBLP:conf/sea2/Tamaki19,DBLP:journals/jea/KorhonenBJ19,DBLP:journals/corr/abs-2012-09153,DBLP:journals/siamcomp/BouchitteT01, DBLP:journals/tcs/BouchitteT02}. A central (exponential time) algorithm for finding an optimal tree decomposition
is the one due to Bouchitté and Todinca~\cite{DBLP:journals/siamcomp/BouchitteT01} (BT algorithm). The BT algorithm has been shown to be highly efficient~\cite{dell_et_al:LIPIcs:2018:8558,DBLP:journals/jea/KorhonenBJ19}, and applicable to a range of cost functions~\cite{DBLP:journals/dam/BodlaenderF05,DBLP:journals/tcs/FuruseY14}.
The first phase of the BT algorithm, and its extensions, employs a subroutine for listing the set of minimal separators of the graph~\cite{DBLP:journals/siamcomp/BouchitteT01}. Consequently, an algorithm that efficiently lists the small minimal separators of a graph can lead to significant performance gains in finding \e{good}, or low-width, tree decompositions. Another parameter of interest is \e{treedepth}~\cite{DBLP:journals/corr/abs-2008-09822,xu_et_al:LIPIcs:2020:13334}, important in a variety of combinatorial optimization problems. Its computation also requires all the minimal separators of the graph of bounded size. Motivated by these tasks, we study the problem of efficiently enumerating small minimal $s,t$-separators.

In any ranked enumeration algorithm, finding the top element is basically an \e{optimization problem}. In our case, there are well known algorithms for finding a \e{minimum-cardinality $s,t$-separator}~\cite{DBLP:journals/jal/HenzingerRG00,DBLP:journals/siamcomp/EvenT75,Chen2022}. For $k>1$, finding the $k$-th ranking item amounts to computing the optimal minimal $s,t$-separator under the restriction that it is not among the first $k-1$ items previously returned. Handling this constraint is the main challenge when designing ranked enumeration and top-$K$ algorithms. 

The technique of Lawler and Murty~\cite{DBLP:journals/siamcomp/LawlerLK80} provides a general framework for ranked enumeration corresponding to discrete optimization problems. The main idea is to reduce a ranked enumeration problem to an optimization problem with constraints~\cite{DBLP:journals/pvldb/GolenbergKS11}.
In the standard approach to applying the Lawler-Murty technique, the algorithm first finds the optimal solution $S$ (e.g., minimum $s,t$-separator), then the subspace of solutions (excluding $S$) is partitioned using \e{inclusion} and \e{exclusion} constraints. The straightforward approach to applying the Lawler-Murty method to ranked enumeration of minimal $s,t$-separators is by solving the following optimization problem: find the smallest minimal $s,t$-separator that excludes a subset $U\subseteq \nodes(G)$, and includes a subset $I\subseteq \nodes(G)$ of vertices. Using this approach, we immediately hit an obstacle. In Appendix~\ref{sec:chordlessPath}, we show that deciding whether 
there exists a minimal $s,t$-separator that includes a distinguished vertex $v\in \nodes(G)$, is NP-complete by reduction from the \textsc{$3$-in-a-path} problem~\cite{DBLP:journals/dam/DerhyP09}. On the other hand, we show that the task of finding the smallest $s,t$- separator that excludes a subset of vertices can be performed in polynomial time. This leads to an algorithm for enumerating the $s,t$-separators in ranked order by size.
Overall, we prove the following.
\begin{theorem}
	\label{thm:fptDelay}
Let $G$ be a finite, undirected graph, and $s,t\in \nodes(G)$ be non-adjacent. The set of minimal $s,t$-separators whose cardinality is at most $k$ can be listed with delay $O(nk4^kT(n,m))$ where $n=|\nodes(G)|$, $m=|\edges(G)|$, and $T(n,m)$ is the time to find a minimum $s,t$-separator of $G$.
\end{theorem}
The problem of finding a minimum $s,t$-separator in an undirected graph can be reduced, by standard techniques~\cite{DBLP:books/daglib/0032640}, to finding a maximum flow from $s$ to $t$. 
Following a sequence of improvements to max-flow algorithms in the past few years~\cite{DBLP:conf/stoc/LiuS20,DBLP:conf/focs/KathuriaLS20,DBLP:conf/stoc/BrandLLSS0W21}, the current best performing algorithm runs in almost linear time $m^{1+o(1)}$~\cite{Chen2022}.
We also prove the following.
\begin{theorem}
	\label{thm:rankednonminimal}
	Let $G$ be a finite, undirected graph, and $s,t\in \nodes(G)$ be non-adjacent. There is an enumeration algorithm that outputs the $s,t$-separators of $G$ in ranked order by cardinality, whose delay is $O(n\cdot T(n,m))$. In particular, our algorithm lists the minimum $s,t$ separators of $G$ with delay  $O(k\cdot T(n,m))$, where $k$ is the cardinality of a minimum $s,t$-separator.
\end{theorem}

The algorithm of Theorem~\ref{thm:rankednonminimal} follows the framework of Lawler and Murty, and is easily implementable (see Figure~\ref{fig:rankedMinSeps}).

{\bf Organization.}
The rest of this paper is organized as follows. Following preliminaries in Section~\ref{sec:Preliminaries}, we prove some useful properties of minimal $s,t$-separators in Section~\ref{sec:charMinlSeps}.
In Section~\ref{sec:theAlg}, we show how small minimal $s,t$-separators can be enumerated in FPT-delay, and in Section~\ref{sec:enumseps} present a polynomial-delay algorithm for enuemrating $s,t$-separtors ranked by size (but without the guarantee of minimality). 
\eat{
 and show how these are applied in the context of the Lawler Murty technique. We present the enumeration algorithm in Section~\ref{sec:theAlg}. Due to space restrictions, most of the technical details are deferred to the Appendix. 
}

\section{Preliminaries and Notation}
\label{sec:Preliminaries}
Let $G$ be an undirected graph with nodes $\nodes(G)$ and edges $\edges(G)$, where $n=|\nodes(G)|$, and $m=|\edges(G)|$. We assume, without loss of generality, that $G$ is connected. 
For $A,B\subseteq \nodes(G)$, we abbreviate $AB\eqdef A\cup B$; for $v\in \nodes(G)$ we abbreviate $vA\eqdef \set{v}\cup A$. Let $v\in V$. We denote by $N_G(v)\eqdef\set{u\in\nodes(G) : (u,v)\in \edges(G)}$ the neighborhood of $v$, and by $N_G[v]\eqdef N_G(v)\cup \set{v}$ the \e{closed} neighborhood of $v$.
For a subset of vertices $T\subseteq \nodes(G)$, we denote by $N_G(T)\eqdef \bigcup_{v\in T}N_G(v){\setminus}T$, and $N_G[T]\eqdef N_G(T)\cup T$.
We denote by $G[T]$ the subgraph of $G$ induced by $T$. Formally, $\nodes(G[T])=T$, and $\edges(G[T])=\set{(u,v)\in \edges(G): \set{u,v}\subseteq T}$. For a subset $S\subseteq \nodes(G)$, we abbreviate $G\sminus S\eqdef G[\nodes(G){\setminus} S]$; for $v\in \nodes(G)$, we abbreviate $G\sminus v\eqdef G\sminus \set{v}$.
We say that $G'$ is a \e{subgraph} of $G$ if it results from $G$ by removing vertices and edges; formally, $\nodes(G')\subseteq \nodes(G)$ and $\edges(G')\subseteq \edges(G)$. In that case, we also say that $G$ is a \e{supergraph} of $G'$.

Let $e=(u,v)\in \edges(G)$ be an edge of $G$. The \e{contraction} of $e$ results in a new graph $G'$, where $u$ and $v$ are identified with a new vertex $w_e$ that is adjacent to $N_G(u)\cup N_G(v)$. Formally,
$\nodes(G')=\nodes(G){\setminus}\set{u,v}\cup \set{w_e}$, and $\edges(G')=\edges(G)\setminus \set{e}\cup \set{(w_e,y):y\in N_G(u)\cup N_G(v)}$. The \e{contraction} of $e$ to vertex $u$ results in the graph $G'$, where $v$ is identified with $u$ that is adjacent to $N_G(u)\cup N_G(v)$. Formally, $\nodes(G')=\nodes(G){\setminus}\set{v}$, and $\edges(G')=\edges(G)\setminus \set{e}\cup \set{(u,y):y\in N_G(v)}$.

Let $u,v \in \nodes(G)$. A \e{simple path} between $u$ and $v$, called a $uv$-path, is a finite sequence of distinct vertices $u=v_1,\dots,v_k=v$ where, for all $i\in [1,k-1]$, $(v_{i},v_{i+1})\in \edges(G)$, and whose ends are $u$ and $v$. A $uv$-path is \e{chordless} or \e{induced} if $(v_i,v_j)\notin \edges(G)$ whenever $|i-j|>1$. 

A subset of vertices $V'\subseteq \nodes(G)$ is called a \e{connected component} of $G$ if $G[V']$ contains a  path between every pair of vertices in $V'$. Let $V_1,V_2\subseteq \nodes(G)$ denote two disjoint vertex subsets of $\nodes(G)$. We say that $V_1$ and $V_2$ are adjacent if there is at least one pair of adjacent vertices $v_1 \in V_1$ and $v_2\in V_2$. We say that there is a path between $V_1$ and $V_2$ if there exist vertices $v_1 \in V_1$ and $v_2\in V_2$ such that there is a path between $v_1$ and $v_2$.

\subsection{Minimal Separators} 
\label{sec:minimalSeparators}
Let $s,t \in 
\nodes(G)$. 
For $X \subseteq \nodes(G)$, we let $\cc(G\sminus X)$ denote the set of connected components of $G\sminus X$. The vertex set $X$ is called a \e{separator} of $G$ if $|\cc(G\sminus X)|\geq 2$, an \e{$s,t$-separator} if $s$ and $t$ are in different connected components of $\cc(G\sminus X)$, and a \e{minimal $s,t$-separator} if no proper subset of $X$ is an $s,t$-separator of $G$. For an $s,t$-separator $X$, we denote by $C_s(G\sminus X)$ and $C_t(G\sminus X)$ the connected components of $\cc(G\sminus X)$ containing $s$ and $t$ respectively.
In other words, $C_s(G\sminus X)=\set{v\in \nodes(G): \text{ there is a path from }s\text{ to }v\text{ in }G\sminus X}$.
\begin{citedlemma}{\cite{DBLP:journals/ijfcs/BerryBC00}}\label{lem:fullComponents}
	An $s,t$-separator $X\subseteq \nodes(G)$ is a minimal $s,t$-separator if and only if $N_G(C_s(G\sminus X))=N_G(C_t(G\sminus X))=X$. \eat{, in which case $C_s(G\sminus X)$ and $C_t(G\sminus X)$ are called \e{full components} of $\cc(G\sminus X)$.}\eat{ 
		there are two connected components $C_s\eqdef C_s(G,X),C_t\eqdef C_t(G,X)\in \cc_G(X)$, such that $s\in C_s$, $t \in C_t$, and $N_G(C_s)=N_G(C_t)=X$; $C_s$ and $C_t$ are called \e{full components} of $\cc_G(X)$.}
\end{citedlemma}
A subset $X\subseteq \nodes(G)$ is a \e{minimal separator} if there exist a pair of vertices $u,v \in \nodes(G)$ such that $X$ is a minimal $u,v$-separator. A connected component $C\in \cc(G\sminus X)$ is called a \e{full component} of $X$ if $N_G(C)=X$. By Lemma~\ref{lem:fullComponents}, $X$ is a minimal $u,v$-separator if and only if the components $C_u(G\sminus X)$ and $C_v(G\sminus X)$ are full components of $X$.
We denote by $\minlsepst{G}$ the set of minimal $s,t$-separators of $G$, and by $\minlsep{}{G}$ the set of minimal separators of $G$.

When $S\in \minlsepst{G}$ where $S\subseteq N_G(s)$, then we say that $S$ is \e{close to} $s$~\cite{DBLP:journals/siamcomp/KloksK98}.
\begin{citedlemma}{\cite{DBLP:journals/siamcomp/KloksK98}}\label{lem:closeTos}
If $s$ and $t$ are non-adjacent, then there exists exactly one minimal $s,t$-separator that is close to $s$.
\end{citedlemma}
\begin{lemma}\label{lem:closeToImplication}
	Let $S\in \minlsepst{G}$ where $S\subseteq N_G(s)$. For every $T\in \minlsepst{G}$, it holds that $C_s(G\sminus S)\subseteq C_s(G\sminus T)$.
\end{lemma}
\begin{proof}
	Since $S\subseteq N_G(s)$, then $S\subseteq N_G(s) \subseteq  T\cup C_s(G\sminus T)$. Hence, $C_s(G\sminus S)  \subseteq C_s(G\sminus T)$.
\end{proof}

\subsection{Important Minimal Separators}
The notion of \e{important separators} has been applied to the design of various fixed-parameter tractable algorithms~\cite{DBLP:conf/wg/Marx11}.
\begin{citeddefinition}{\cite{DBLP:books/sp/CyganFKLMPPS15}}
	Let $S\subseteq \nodes(G)$. We say that $S$ is an \e{important $s,t$-separator} if $S\in \minlsepst{G}$, and for any other $S'\in \minlsepst{G}$ it holds that:
	\begin{align*}
		C_s(G\sminus S') \subset C_s(G\sminus S) \Longrightarrow |S'|>|S| 
		\end{align*}
\end{citeddefinition}
In what follows, we denote by $\impsepst{G}$ the set of important $s,t$-separators, and by $\impsepstk{G}$ the set of important $s,t$-separators whose size is at most $k$.
\begin{citedtheorem}{\cite{DBLP:books/sp/CyganFKLMPPS15}}
Let $T$ be an important $s,t$-separator of smallest size. Then the following holds:
\begin{enumerate}
	\item $G$ contains exactly one important separator of minimum size. That is, $T$ is unique. 
	\item $T$ can be found in time $O(nT(n,m))$ where $n=|\nodes(G)|$, $m=|\edges(G)|$, and $T(n,m)$ is the time to find a minimum $s,t$-separator in $G$.
\end{enumerate}
\end{citedtheorem}

\begin{citedtheorem}{\cite{DBLP:books/sp/CyganFKLMPPS15}}
	\label{thm:importantSepsEnum}
	There are at most $4^k$ important $s,t$-separators of $G$ whose size is at most $k$, and there is an algorithm that outputs them in total time $O(n\cdot T(n,m)\cdot 4^k)$.
	\eat{There is an algorithm which enumerates all important $s,t$-separators of $G$ whose size is at most $k$, that runs in time $O(n\cdot T(n,m)\cdot 4^k)$.}
\end{citedtheorem}

\subsection{Minimum Separators}
\label{sec:minseps}
A subset $S \subseteq \nodes(G)$ is a \e{minimum $s,t$-separator} of $G$ if $|S'|\geq |S|$ for every other $s,t$-separator $S'$.  We denote by $\kappa_{s,t}(G)$ the size of a minimum $s,t$-separator of $G$, and by $\minsepst{G}$ the set of all minimum $s,t$-separators of $G$; $\kappa_{s,t}(G)$ is called the $s,t$-\e{connectivity} of $G$.\eat{
Similarly, for $A,B\subseteq \nodes(G)$ that are disjoint and non-adjacent, we say that a subset $X\subseteq \nodes(G){\setminus} AB$ is a minimum $A,B$-separator if, for every $A,B$-separator $S$, it holds that $|X|\leq |S|$. We denote by $\minsep_{A,B}(G)$ the set of minimum $A,B$-separators, and by $\kappa_{A,B}(G)$ their size. We say that $X$ is a minimum, safe $A,B$-separator if it is a safe $A,B$-separator of minimum size.}
Finding a minimum $s,t$-separator can be reduced, by standard techniques~\cite{DBLP:books/daglib/0032640}, to the problem of finding a\eat{ minimum $s,t$-edge-cut, which is equivalent to finding} a maximum flow in the graph~\cite{10.5555/1942094}. Currently, the fastest known algorithm for max-flow runs in almost linear time $m^{1+o(1)}$~\cite{Chen2022}.

\begin{citedtheorem}{Menger~\cite{DBLP:books/daglib/0030488}}
	\label{thm:Menger}
	Let $G$ be an undirected graph and $s,t \in \nodes(G)$. Then the minimum number of vertices separating $s$ from $t$ in $G$ is equal to the maximum number of internally vertex-disjoint $s,t$-paths in $G$.
\end{citedtheorem}

\eat{
Lemma~\ref{lem:vertexInclude} below defines a simple procedure for testing whether a distinguished vertex $v\in \nodes(G)$ belongs to some minimum $s,t$-separator. The lemma is crucial for the ranked enumeration algorithm, and its proof is deferred to Appendix~\ref{sec:minsepsvertexsets}.
\def\vertexIncludeLem{
		Let $v\in \nodes(G)$. There exists a minimum $s,t$-separator $S\in \minsepst{G}$ that contains $v$ if and only if $\kappa_{s,t}(G\sminus v)=\kappa_{s,t}(G)-1$.
}
\begin{lemma}
	\label{lem:vertexInclude}
\vertexIncludeLem
\end{lemma}
}

\section{A required Characterization of Minimal $s,t$-separators}
\label{sec:charMinlSeps}
In this section, we prove three useful lemmas concerning minimal $s,t$-separators that will be instrumental for the enumeration algorithm.
\eat{
\begin{lemma}
	\label{lem:belongtoCt}
	Let $S\in \minlsepst{G}$ such that $S\subseteq N_G(s)$, and let $u\in \nodes(G)$. If $u\in C_s(G\sminus S)$ then $u\in C_s(G\sminus T)$ for every $T\in \minlsepst{G}$.
\end{lemma}
\begin{proof}
	Since $u\notin C_s(G\sminus S)$, then every path from $u$ to $s$ passes through a vertex in $S$. Now, let $T\in \minlsepst{G}{\setminus} \set{S}$. Since $S\subseteq N_G(t)$, then $S\subseteq T\cup C_t(G\sminus T)$. Therefore, every path from a vertex in $S$ to $s$ passes through a vertex in $T$. Consequently, every path from $u$ to $s$, which passes through a vertex in $S$, must also pass through a vertex in $T$. Therefore, $u\notin C_s(G\sminus T)$.
\end{proof}
}
\begin{lemma}
	\label{lem:contract}
Let $v\in N_G(s)$, and let $G'$ denote the graph where $\nodes(G')=\nodes(G)$ and $\edges(G')=\edges(G)\cup \set{(s,y):y\in N_G(v)}$. Then $\minlsepst{G'}=\set{S\in \minlsepst{G}: v\notin S}$. 
\end{lemma}
\begin{proof}
Let $S\in \minlsepst{G}$ such that $v\notin S$. Since $v\in N_G(s)$, then $v\in C_s(G\sminus S)$, and hence $N_G(v)\subseteq S\cup C_s(G\sminus S)$. Therefore, $S$ is an $s,t$-separator in $G'$ as well. Since $\edges(G')\supseteq \edges(G)$, then $S\in \minlsepst{G'}$. Now, let $T\in \minlsepst{G'}$. Since $N_{G'}(v) \subseteq N_{G'}(s)$, and $v\in N_{G'}(s)$, then $v\in C_s(G'\sminus T)$, and in particular, $v\notin T$. Since $\edges(G')\supseteq \edges(G)$, then $T$ separates $s$ from $t$ in $G$. Therefore, there is a subset $T'\subseteq T$ such that $T'\in \minlsepst{G}$. Since $v\notin T$, then $v\notin T'$. Hence, $T'\in \set{S\in \minlsepst{G}: v\notin S}$. By the previous, $T'\in \minlsepst{G'}$. If $T'\subset T$, we arrive at a contradiction that $T\in \minlsepst{G'}$. Therefore, $T'=T$, and hence $T\in \minlsepst{G}$ as required.
\end{proof}
\begin{lemma}
	\label{lem:parallel1}
Let $S,T\in \minlsepst{G}$. Then:
\begin{align*}
	C_s(G\sminus S) \subseteq C_s(G\sminus T) \text{ if and only if } T\subseteq S\cup C_t(G\sminus S)
\end{align*}
\end{lemma}
\begin{proof}
If $T\subseteq S\cup C_t(G\sminus S)$, then by definition $T\cap C_s(G\sminus S)=\emptyset$. Therefore, $C_s(G\sminus S)$ remains connected in $G\sminus T$. This means that $C_s(G\sminus S)\subseteq C_s(G\sminus T)$.  

Now, suppose that $C_s(G\sminus S) \subseteq C_s(G\sminus T)$. By Lemma~\ref{lem:fullComponents}, it holds that $S=N_G(C_s(G\sminus S))$. Since $C_s(G\sminus S) \subseteq C_s(G\sminus T)$, then $S=N_G(C_s(G\sminus S))\subseteq T\cup C_s(G\sminus T)$. Since $S\subseteq T\cup C_s(G\sminus T)$ then by definition it holds that $S\cap C_t(G\sminus T)=\emptyset$. This, in turn, implies that $C_t(G\sminus T)$ remains connected in $G\sminus S$. In particular, we have that $C_t(G\sminus T)\subseteq C_t(G\sminus S)$. By Lemma~\ref{lem:fullComponents}, it holds that $T=N_G(C_t(G\sminus T))$. Since $C_t(G\sminus T)\subseteq C_t(G\sminus S)$, then $T=N_G(C_t(G\sminus T))\subseteq S\cup C_t(G\sminus S)$.
\end{proof}

\begin{lemma}
	\label{lem:Hgraph}
	Let $S\in \minlsepst{G}$, and let $H_S$ be the graph that results from $G$ by adding all edges from $s$ to $S$. That is, $\edges(H_S)=\edges(G)\cup \set{(s,v):v\in S}$. Then:
	\begin{align*}
		\minlsepst{H_S}=\set{Q\in \minlsepst{G}: Q\subseteq S\cup C_t(G\sminus S)}
	\end{align*}
\end{lemma}
\begin{proof}
	Let $Q\in \minlsepst{G}$ where $Q\subseteq S\cup C_t(G\sminus S)$. Since $Q\cap C_s(G\sminus S)=\emptyset$, then $C_s(G\sminus S)$ remains connected in $G\sminus Q$. Therefore, $C_s(G\sminus S) \subseteq C_s(G\sminus Q)$. By Lemma~\ref{lem:fullComponents}, $S=N_G(C_s(G\sminus S))$. Since $C_s(G\sminus S) \subseteq C_s(G\sminus Q)$, then $S=N_G(C_s(G\sminus S))\subseteq C_s(G\sminus Q)\cup Q$. In particular, $S\cap C_t(G\sminus Q)=\emptyset$. Consequently, $Q$ separates $C_t(G\sminus Q)$ from $s$ in $H_S$ as well. That is, $Q$ is an $s,t$-separator in $H_S$. Since $\edges(H_S)\supseteq \edges(G)$, then $Q\in \minlsepst{H_S}$.
	
	Let $T\in \minlsepst{H_S}$. By construction, $S\in \minlsepst{H_S}$ where $S\subseteq N_H(s)$. By Lemma~\ref{lem:closeToImplication}, it holds that $C_s(H_S\sminus S)\subseteq C_s(H_S\sminus T)$. By Lemma~\ref{lem:parallel1}, it holds that $T\subseteq S\cup C_t(H_S\sminus S)$. Since, by construction, $C_t(H_S\sminus S)=C_t(G\sminus S)$, we get that $T\subseteq S\cup C_t(G\sminus S)$.
\end{proof}
\def\lemaddEdges{
	Let $S\in \minlsepst{G}$ and let $\set{(u_1,v_1),\dots,(u_k,v_k)}$ be a non-empty subset of non-adjacent vertices of $G[S]$. Let $G'$ be the graph where $\nodes(G')=\nodes(G)$ and $\edges(G')=\edges(G)\cup \bigcup_{i=1}^k(u_i,v_i)$. Then $	\minlsepst{G'}=\bigcap_{i=1}^k\minlsepG{G}{34}{u_i,v_i}$.\eat{
	\begin{equation}
		\minlsepst{G'}=\bigcap_{i=1}^k\left(\minlsepG{G}{3}{u_i,v_i}\cup \minlsepG{G}{4}{u_i,v_i}\right)
	\end{equation}	
}
}

\section{Listing small, minimal $s,t$-separators in FPT-delay}
\label{sec:theAlg}
Given two subsets $C_1,C_2 \subseteq \nodes(G)$, we say that they are \e{incomparable} if $C_1\not\subseteq C_2$ and $C_2\not\subseteq C_1$. 
We define a partial order $\preceq$ over the members of $\minlsepst{G}$. 
We say that $S \prec T$ if and only if $C_s(G\sminus S)\subset C_s(G\sminus T)$. Our enumeration algorithm will output all minimal $s,t$-separators whose size is at most $k$ according to the order $\prec$.

\begin{algserieswide}
	{H}{Algorithm for listing the minimal $s,t$-separators of $G$ whose size is at most $k$. \label{fig:rankedMinSeps}}
	\begin{insidealgwide}{SmallMinimalSeps}{$G$, $s$, $t$, $k$}
		\IF{$(s,t)\in \edges(G)$}
			\STATE Print $\bot$
			\RETURN
		\ENDIF
		\STATE $Q\gets \algname{PriorityQueue}(\preceq)$
		\STATE Compute $\impsepstk{G}$ \COMMENT{Takes time $O(n\cdot T(n,m)\cdot 4^k)$ (see Theorem~\ref{thm:importantSepsEnum})}
		\FOR{$S\in \impsepstk{G}$}
			\STATE $Q.\texttt{push}(S)$ \label{line:pushImportantG}
		\ENDFOR
		\WHILE{$Q$ is not empty}
			\STATE $S\gets Q.\texttt{pop}()$ \label{line:popQ}
			\STATE Print $S$ \label{line:printS}
			\STATE Let $H_S$ be the graph where $\nodes(H_S)=\nodes(G)$, and $\edges(H_S)=\edges(G)\cup \set{(s,v):v\in S}$ \label{line:generateHS}
			\FOR{$v\in S$}
				\STATE Let $H_S^v$ be the graph that results from $H_S$ by adding all edges $\set{(s,y):y\in N_{H_S}(v)}$. \label{line:HContractv}
				\STATE Compute $\impsepstk{H_S^v}$ 
				\FOR{$T\in \impsepstk{H_S^v}$}
				\IF{$T\notin Q$}
					\STATE $Q.\texttt{push}(T)$ \label{line:pushQ}
				\ENDIF
				\ENDFOR
			\ENDFOR
		\ENDWHILE	
	\end{insidealgwide}
\end{algserieswide}

\subsection{Proof of Correctness}
\label{sec:algCorrectness}
We begin by establishing that a subset $S\subseteq \nodes(G)$ is printed by the algorithm if and only if $S\in \minlsepst{G}$, and $|S|\leq k$. Then, we analyze the runtime. 
\begin{theorem}
If $S\subseteq \nodes(G)$ is printed by the algorithm, then $S\in \minlsepst{G}$ and $|S|\leq k$.
\end{theorem}
\begin{proof}
	We first note that every subset of vertices inserted into the queue (in lines~\ref{line:pushImportantG} and~\ref{line:pushQ}) has cardinality at most $k$, and separates $s$ from $t$. Therefore, we only need to show that every subset of vertices printed by the algorithm belongs to $\minlsepst{G}$.
	
Suppose, by way of contradiction, that this is not the case, and let $T\subseteq \nodes(G)$ be the first subset of vertices printed by the algorithm where $T\notin \minlsepst{G}$. Then $T$ must be inserted into the queue in line~\ref{line:pushQ}. Consider the set $S$ that was printed before $T$ is inserted into the queue. By our assumption $S\in \minlsepst{G}$. Therefore, $T\in \impsepstk{H_S^v}$, where $v\in S$. By Lemma~\ref{lem:Hgraph}, $\minlsepst{H_S}\subseteq \minlsepst{G}$. Since $v\in N_{H_S}(s)$, and $H_S^v$ is the graph that results from $H_S$ by adding all edges $\set{(s,y):y\in N_{H_S}(v)}$, by Lemma~\ref{lem:contract}, it holds that $\minlsepst{H_S^v}\subseteq \minlsepst{H_S}\subseteq  \minlsepst{G}$. 
Since $T\in \impsepstk{H_S^v}\subseteq \minlsepst{H_S^v}$, we get that $T\in \minlsepst{G}$, which brings us to a contradiction.
\end{proof}
We now prove that every minimal separator $T\in \minlsepst{G}$, where $|T|\leq k$, is printed by the algorithm.
\begin{lemma}
	\label{lem:mainCorrectness}
	Let $T\in \minlsepst{G}$ such that $|T|\leq k$. There exists a $S\in \impsepstk{G}$ such that $C_s(G\sminus S)\subseteq C_s(G\sminus T)$.
\end{lemma}
\begin{proof}
	If $T\in \impsepstk{G}$, then the claim is immediate. We prove the claim by induction on $|C_s(G\sminus T)|$. If $|C_s(G\sminus T)|=1$, then clearly $T\subseteq N_G(s)$. By Lemma~\ref{lem:closeTos}, $T$ is the unique minimal $s,t$-separator that is close to $s$, and hence $T\in \impsepstk{G}$, and the claim follows.
	So, we assume the claim holds for all $T\in \minlsepst{G}$, where $|T|\leq k$, and $1\leq |C_s(G\sminus T)|\leq \ell$. Let $T\in \minlsepst{G}$, where $|T|\leq k$, and $|C_s(G\sminus T)|= \ell+1$. Since $T\notin \impsepstk{G}$, then by definition, there exists a $T'\in \minlsepst{G}$ such that $|T'|\leq |T|\leq k$, and $C_s(G\sminus T')\subset C_s(G\sminus T)$.  \eat{By Lemma~\ref{lem:parallel1} it holds that $T\subseteq T'\cup C_t(G\sminus T')$.} If $T'\in  \impsepstk{G}$, then we are done. 
	Otherwise, observe that $|C_s(G\sminus T')|<|C_s(G\sminus T)|=\ell+1$. Therefore, $|C_s(G\sminus T')|\leq \ell$, and by the induction hypothesis there exists a $S\in \impsepstk{G}$ such that $C_s(G\sminus S)\subseteq C_s(G\sminus T')$. Consequently, we have that $C_s(G\sminus S)\subseteq C_s(G\sminus T')\subset C_s(G\sminus T)$, and thus $C_s(G\sminus S)\subseteq C_s(G\sminus T)$, where $S\in  \impsepstk{G}$. This proves the lemma.
\end{proof}

\begin{corollary}
		\label{corr:mainCorrectness}
			Let $T\in \minlsepst{G}$ such that $|T|\leq k$. There exists a $S\in \impsepstk{G}$ such that $T\in S\cup C_t(G\sminus S)$.
\end{corollary}
\begin{proof}
By Lemma~\ref{lem:mainCorrectness}, there exists a $S\in \impsepstk{G}$ such that $C_s(G\sminus S)\subseteq C_s(G\sminus T)$. By Lemma~\ref{lem:parallel1}, it holds that $T\in S\cup C_t(G\sminus S)$.
\end{proof}

\begin{theorem}
	\label{thm:mainCorrectness}
		Let $T\in \minlsepst{G}$ such that $|T|\leq k$. Then $T$ is printed by Algorithm \algname{SmallMinimalSeps} in Figure~\ref{fig:rankedMinSeps}.
\end{theorem}
\begin{proof}
Suppose that $T$ is not printed by the algorithm. Let $T'\in \minlsepst{G}$ be the largest minimal $s,t$-separator, with respect to $\prec$, that is printed by the algorithm, such that $T'\preceq T$. In other words, there does not exist a $T''\in \minlsepst{G}$, that is printed by the algorithm where $T'\prec T''\preceq T$. By Lemma~\ref{lem:mainCorrectness}, such a separator $T'$ must exist.

Since  $C_s(G\sminus T')\subseteq C_s(G\sminus T)$, then by Lemma~\ref{lem:parallel1}, it holds that $T\in T'\cup C_t(G\sminus T')$. By Lemma~\ref{lem:Hgraph}, it holds that $T\in \minlsepst{H_{T'}}$. Consider what happens when $T'$ is popped from the queue in line~\ref{line:popQ}, and the graph $H_{T'}$ is generated in line~\ref{line:generateHS}. Since $T\neq T'$ (we assume that $T$ is not printed), $T'\subseteq N_{H_{T'}}(s)$, and $T\in \minlsepst{H_{T'}}$, then there exists a vertex $v\in T'$, such that $T\in \minlsepst{H_{T'}^v}$ (see line~\ref{line:HContractv}). If $T\in \impsepstk{H_{T'}^v}$, then $T$ is pushed into the queue in line~\ref{line:pushQ}, and will therefore be printed. Otherwise, by Lemma~\ref{lem:mainCorrectness}, there exists an $S\in \impsepstk{H_{T'}^v}$, such that $C_s(H_{T'}^v\sminus S)\subseteq C_s(H_{T'}^v\sminus T)$. By construction, we have that  $C_s(H_{T'}\sminus T')\subseteq  C_s(H_{T'}^v\sminus S)\subseteq C_s(H_{T'}\sminus T)$. Since $S$ is pushed into the queue in line~\ref{line:pushQ}, then it will be printed by the algorithm in line~\ref{line:printS}. By Theorem, we have that $S\in \minlsepst{G}$, where $|S|\leq k$ and where $C_s(G\sminus T')\subseteq C_s(G\sminus S)\subseteq C_s(G\sminus T)$ is printed by the algorithm, contradicting our assumption that $T'$ is maximal with respect to the partial order $\prec$ that is printed before $T$.
\end{proof}

\begin{theorem}
	\label{thm:runtime}
	The delay between the printing of consecutive minimal $s,t$-separators whose size is at most $k$ is $O(n\cdot k\cdot T(n,m)\cdot 4^k)$, where $n=|\nodes(G)|$, $m=|\edges(G)|$, and $T(n,m)$ is the time to find a minimum $s,t$-separator in $G$.
\end{theorem}
\begin{proof}
	The size of the queue can be at most $n^k$. We make the standard assumption that the queue allows logarithmic insertion and extraction. Hence, insertion and extraction from the queue take time $O(k\log n)$. Therefore, the runtime of the loop in lines is:
	$$
	O\left(k\cdot \left(n+n\cdot T(n,m)\cdot 4^k + k\cdot 4^k\log n\right)\right)
	$$
	Overall, the delay is $O(n\cdot k\cdot T(n,m)\cdot 4^k)$.
\end{proof}
\section{Enumeration Algorithms for $s,t$-Separators}
\label{sec:enumseps}
In this section, we develop an enumeration algorithm that returns all (not necessarily minimal) $s,t$-separators of $G$ in ranked order by cardinality, and an enumeration algorithm that returns only the minimum-cardinality $s,t$-separators of $G$ (i.e., $\minsep_{s,t}(G)$). 
We characterize vertices included in minimum-cardinality $s,t$-separators in Section~\ref{sec:includeMinSeps}, and vertices excluded from any minimal $s,t$-separator (and in-particular, any minimum $s,t$-separator) in Section~\ref{sec:excludeU}.
The resulting simple algorithms are then presented in Section~\ref{sec:enumAlgs}.
\subsection{Characterizations of Vertices Included In Minimum $s,t$-Separators}
\label{sec:includeMinSeps}
In this section, we show that there is a polynomial-time algorithm, that given a subset of vertices $I\subseteq \nodes(G)$, returns a minimum $s,t$-separator that contains $I$, if one exists.
\begin{lemma}
	\label{lem:vertexInclude}
	Let $v\in \nodes(G)$. There exists a minimum $s,t$-separator $S\in \minsep_{s,t}(G)$ that contains $v$ if and only if $\kappa_{s,t}(G\sminus v)=\kappa_{s,t}(G)-1$.
\end{lemma}
\begin{proof}
	Let $\kappa_{s,t}(G)=k$. Let $S\in \minsep_{s,t}(G)$ be such that $v \in S$. Consider the graph $G'\eqdef G\sminus v$. Removing a single vertex can decrease the $s,t$-connectivity of a graph by at most $1$. Therefore, $\kappa_{s,t}(G)-1 \leq \kappa_{s,t}(G')$. On the other hand, $S\setminus \set{v}$ is an $s,t$-separator in $G'$, and hence $\kappa_{s,t}(G')\leq |S|-1=\kappa_{s,t}(G)-1$. Overall, we get that $\kappa_{s,t}(G')=\kappa_{s,t}(G-v)=\kappa_{s,t}(G)-1$.
	
	Now, suppose that $\kappa_{s,t}(G)=k$ and $\kappa_{s,t}(G\sminus v)=k-1$. Take $S\in \minsep_{s,t}(G\sminus v)$. By our assumption $|S|=k-1$. Since $S$ separates $s$ from $t$ in $G\sminus v$, then clearly $S\cup \set{v}$ separates $s$ from $t$ in $G$. Since $|S\cup \set{v}|=k=\kappa_{s,t}(G)$, then $S\in \minsep_{s,t}(G)$ where $v\in S$.
\end{proof}

\begin{corollary}
	\label{corr:vertexSetInclude}
	Let $I \subseteq \nodes (G)$, and let $k_I=|I|$. There exists a minimum separator $S \in \minsep_{st}(G)$ that contains $I$ if and only if $\kappa_{s,t}(G\sminus I)=\kappa_{s,t}(G)-k_I$.
\end{corollary}
\begin{proof}
	Let $\kappa_{s,t}(G)=k$. We prove by induction on $k_I$, the size of $I$. The case where $I=\emptyset$ is trivial, and the case where $|I|=1$ is proved in Lemma~\ref{lem:vertexInclude}. So assume the claim holds for all sets $I$ where $0\leq |I|\leq \ell$, and we prove for the case where $|I|=\ell+1$.
	
	Let $x\in I$. By Lemma~\ref{lem:vertexInclude}, there exists a minimum separator $S \in \minsep_{s,t}(G)$ that contains $x$ if and only if $\kappa_{s,t}(G\sminus x)=\kappa_{s,t}(G)-1$. Therefore, if $\kappa_{s,t}(G\sminus x)>\kappa_{s,t}(G)-1$, there does not exist a minimum separator that includes $x$, and hence, there does not exist one that contains $I$.
	Otherwise, let $H\eqdef G\sminus x$, and hence $\kappa_{s,t}(H)=k-1$. Let $Y\eqdef I\setminus \set{x}$. Since $|Y|\leq \ell$, then we can apply the induction hypothesis. There is a minimum separator $S \in \minsep_{s,t}(H)$ that contains $Y$ if and only if $\kappa_{s,t}(H\sminus Y)=\kappa_{s,t}(H)-|Y|=\kappa_{s,t}(H)-\ell$.
	Since $H=G\sminus x$, we get that
	\begin{align*}
		\kappa_{s,t}(G\sminus I)&=\kappa_{s,t}(G\sminus (Y\cup \set{x}))\\
		&=\kappa_{s,t}(H\sminus Y)=\kappa_{s,t}(H)-|Y|=\kappa_{s,t}(H)-\ell\\
		&=\kappa_{s,t}(G)-1-\ell=\kappa_{s,t}(G)-(\ell+1)
	\end{align*}
\end{proof}

\begin{proposition}
	\label{prop:inclusionPTimeDecision}
	Let $I\subseteq \nodes(G)$. We can determine whether there is a minimum $s,t$-separator $S \in \minsep_{s,t}(G)$ that contains $I$, and return one if exists, in time $O(n+T(n,m))$, where $T(n,m)$ is the time to find a minimum $s,t$-separator.
\end{proposition}

\subsection{Excluding Vertices from Minimal Separators}
\label{sec:excludeU}
In this section, we characterize minimal $s,t$-separators that exclude a subset $U \subseteq \nodes(G)$ of vertices. 
We define some required notation. 
Let $U \subseteq \nodes(G)$. We denote by $\minlsepEst{G}{U}$ the set of minimal $s,t$-separators that exclude $U$. Formally:
\begin{equation}
	\label{eq:MinSepGIU}
	\minlsepEst{G}{U}\eqdef \set{S \in \minlsepst{G}: S \subseteq \nodes(G)\setminus U}
\end{equation}
For a single vertex $u\in \nodes(G)$, we denote by $\minlsepEst{G}{u}\eqdef \set{S \in \minlsepst{G}:u\notin S}$.
We denote by $\kappa_{s,t}(G,\comp{U})$ the minimum size of any minimal $s,t$-separator that excludes $U$. Formally:
\begin{equation}
	\label{eq:kappaIncludeExclude}
	\kappa_{s,t}(G,\comp{U})\eqdef\min\set{|S| : S \in \minlsepEst{G}{U}}
\end{equation}
We denote by $\minsep_{s,t}(G,\comp{U})$ the subset of $\minlsepEst{G}{U}$ that have the smallest cardinality. Formally:
\begin{equation}
	\label{eq:MinimumIncludeExcludeFixed}
	\minsep_{s,t}(G,\comp{U})\eqdef\set{S\in \minlsepEst{G}{U} : |S|=\kappa_{s,t}(G,\comp{U})}
\end{equation}
Let $u \in \nodes(G)$; we denote by $\sat(G,u)$ the graph that results by adding edges between all vertices in $N_G[u]$. In other words, $\sat(G,u)$ is the graph where the set $N_G[u]$ has been \e{saturated}, and forms a clique.
For a set of vertices $U \subseteq \nodes(G)$, we denote by $\sat(G,U)$ the graph that results by adding edges between all vertices in $N_G[u]$ for all $u\in U$. Formally, $\nodes(\sat(G,U))=\nodes(G)$ and 
\begin{equation}
	\label{eq:ESatG}
	\edges(\sat(G,U))\eqdef\edges(G)\cup \bigcup_{u\in U}\set{(x,y):x,y \in N_G[u]}
\end{equation}
\eat{
	For this section, we recall the definitions of $\minlsepEst{G}{U}$ (see~\eqref{eq:MinSepGIU}), $\minsep_{st}(G,\comp{U})$ (see~\eqref{eq:MinimumIncludeExcludeFixed}), and $\sat(G,U)$ (see~\eqref{eq:ESatG}).} We prove that $\minlsepEst{G}{U}=\minlsepst{\sat(G,U)}$.
We proceed by a series of lemmas.

\begin{lemma}
	\label{lem:excludeNoClique}
	Let $u\in \nodes(G)$ such that $N_G[u]$ forms a clique. Then $u\notin S$ for every $S\in \minlsepst{G}$.
\end{lemma}
\begin{proof}
	Let $S\in \minlsepst{G}$. By Lemma~\ref{lem:fullComponents}, $G\sminus S$ contains two full components $C_s(G\sminus S)$ and $C_t(G\sminus S)$ containing $s$ and $t$ respectively, such that $S=N_G(C_s(G\sminus S))=N_G(C_t(G\sminus S))$. Therefore, if $u \in S$, then it has two neighbors $v_1 \in C_s(G\sminus S)$ and $v_2 \in C_t(G\sminus S)$ that are connected by an edge (because $N_G[u]$ is a clique). But then, there is an $s,t$-path in $G\sminus S$ that avoids $S$, which contradicts the fact that $S$ is an $s,t$-separator.
\end{proof}

\begin{lemma}
	\label{lem:exclude_u}
	For every $S \in \minlsepEst{G}{u}$, there exists a connected component $C_u\in \cc_G(S)$ such that $N_G[u] \subseteq C_u \cup S$.
\end{lemma}
\begin{proof}	
	Let $C_u\in \cc_G(S)$ be the connected component that contains $u$. Such a component must exist because $u \notin S$. If $N_G(u) \not\subseteq C_u\cup S$, then there exists a vertex $v \in N_G(u)$ that resides in a connected component $C_v \in \cc_G(S)$ distinct from $C_u$. But this is a contradiction because, by definition, $(u,v)\in \edges(G)$. Hence, $C_v=C_u$, and this proves the claim.
\end{proof}

\begin{lemma}
	\label{lem:MinlSepIU}
	Let $u \in \nodes(G)$. Then $\minlsepEst{G}{u}=\minlsepst{\sat(G,\set{u})}$.
\end{lemma}
\begin{proof}
	Let $S \in \minlsepEst{G}{u}$. Let $C_u \in \cc_G(S)$ denote the connected component containing $u$ in $G\sminus S$.
	By Lemma~\ref{lem:exclude_u}, $N_G[u] \subseteq C_u\cup S$. Therefore, no added edge in $\edges(\sat(G,\set{u}))\setminus \edges(G)$ connects vertices in distinct connected components in $\cc_G(S)$. Hence, $S$ separates $s$ and $t$ also in $\sat(G,\set{u})$. Since the addition of edges cannot eliminate any path between $s$ and $t$, we get that $S$ is a minimal $s,t$-separator also in $\sat(G,\set{u})$. \eat{The claim for the entire set $U$ follows by induction.}

	Now, let $S \in \minlsepst{\sat(G,\set{u})}$. Hence, $N_G[u]$ is a clique in $\sat(G,\set{u})$. By Lemma~\ref{lem:excludeNoClique}, $u \notin S$.
	Since $G$ is a subgraph of $\sat(G,\set{u})$, then if $S$ separates $s$ from $t$ in $\sat(G,\set{u})$, it must separate $s$ from $t$ in $G$. Hence, $S$ is an $s,t$-separator in $G$ where $u\notin S$. 
	It is left to show that $S$ is a \e{minimal} $s,t$-separator in $G$. To that end, we show that the connected components $C_s,C_t \in \cc_{\sat(G,\set{u})}(S)$, containing $s$ and $t$ respectively, are full connected components of $S$ also in $G$. That is, we show that $S=N_G(C_s)=N_G(C_t)$. By Lemma~\ref{lem:fullComponents}, this proves that $S\in \minlsepEst{G}{u}\subseteq \minlsepst{G}$.
	
	Denote by $D_s,D_t\in \cc_G(S)$ the connected components containing $s$ and $t$ respectively in $G\sminus S$.
	Since $G[D_s]$ ($G[D_t]$) is connected, $D_s\cap S=\emptyset$ ($D_t\cap S=\emptyset$), and $s\in D_s$ ($t\in D_t$), then $D_s \subseteq C_s$ ($D_t \subseteq C_t$). We now prove that $C_s \subseteq D_s$. We first consider the case where $u \notin C_s$. Hence, by definition of connected component of $G\sminus S$, we have that $N_G[u]\cap C_s=\emptyset$. Since the only added edges are between vertices in $N_G(u)$, then $\edges(\sat(G,u)[C_s])=\edges(G[C_s])$. Therefore $C_s$ is a connected component containing $s$ also in $G\sminus S$, thus $C_s\subseteq D_s$, and $C_s=D_s$. Since $N_G[u]\cap C_s=\emptyset$, then $N_G(C_s)=N_{\sat(G,\set{u})}(C_s)=S$ as required. 
	
	We now consider the case where $u\in C_s$, and
	suppose, by way of contradiction, that \eat{$D_s \subset C_s$} $C_s \not\subseteq D_s$. Let $v \in C_s \setminus D_s$. This means that there is a path from $s$ to $v$ in $\sat(G,\set{u})$ that avoids $S$. Let $P$ denote the shortest such path. Then $P$ passes through a single edge $(y,w)\in \edges(\sat(G,u))\setminus \edges(G)$. In other words, there is a path $P_{vy}$ from $v$ to $y$ in $G$ that avoids $S$, and a path $P_{sw}$ from $s$ to $w$ in $G$ that avoids $S$. In particular, $\nodes(P_{sw})\subseteq D_s$.
	By construction, $\set{y,w}\subseteq N_G(u)$. \eat{Since the path $P$ from $s$ to $v$ lies entirely in $G[D_s]$, then $\set{y,w}\subseteq D_s$. Hence, $u\in C_s$, and $\set{y,w}\subseteq N_G(u)$.}
	Since $\nodes(P_{ws})\cap S=\emptyset$,  $w\in N_G(u)$, and $\set{u,w,y}\cap S=\emptyset$, this means that $\set{u,w,y}\subseteq D_s$. 
	But this means that the path $P_{vy}uP_{ws}$ is contained in $G$, and avoids $S$. Consequently, $v\in D_s$, and we arrive at a contradiction.
	Hence, $D_s=C_s$. Since $u\in C_s$, we get that $N_G(C_s)=N_{\sat(G,\set{u})}(C_s)$, making $C_s$ a full connected component of $S$ also in $G$. This completes the proof.
	
	\eat{
		Suppose, by way of contradiction, that there is a vertex $w \in S$, such that $N_G(w)\cap C_s=\emptyset$ or $N_G(w)\cap C_t=\emptyset$. Wlog suppose that $N_G(w)\cap C_s(S)=\emptyset$. Since $S$ is a minimal separator in $\sat(G,U)$, then (by Lemma~\ref{lem:fullComponents}) there is a vertex $v\in C_s$ such that $(w,v)\in \edges(\sat(G,U))\setminus \edges(G)$.
		This, in turn, means that $w$ and $v$ have a common neighbor $u\in U$ in $G$. That is, $(u,v),(u,w)\in \edges(G)$. But since $v\in C_s$, $u\notin S$, and $(u,v)\in \edges(G)$, then $u\in C_s$. But this means that $u\in N_G(w)\cap C_s$, which is a contradiction to our assumption that $N_G(w)\cap C_s=\emptyset$. Hence, $N_G(w)\cap C_s\neq \emptyset$. Symmetrically, $N_G(w)\cap C_t\neq \emptyset$. Thus, $C_s$ and $C_t$ are full components of $S$ in $G$, containing $s$ and $t$ respectively, and by Lemma~\ref{lem:fullComponents}, $S\in \minlsepEst{G}{U}$.
	}
\end{proof}

\begin{theorem}
	\label{thm:MinlSepIU}
	Let $U \subseteq \nodes(G)$. Then $\minlsepEst{G}{U}=\minlsepst{\sat(G,U)}$.
\end{theorem}
\begin{proof}
	This follows from Lemma~\ref{lem:MinlSepIU} by a simple induction on $|U|$. 
\end{proof}

\begin{corollary}
	\label{cor:IUThm}
	The following holds for every $U\subseteq \nodes(G)$. 
	$\minsep_{st}(G,\comp{U})=\minsep_{st}(\sat(G,U))$.
\end{corollary}
\begin{proof}
	In Theorem~\ref{thm:MinlSepIU}, we have shown that for any $U \subseteq \nodes(G)$, it holds that $\minlsepEst{G}{U}=\minlsepst{\sat(G,U)}$. 
	Let $0\leq k\leq n$ be an integer, and $\minlsepEst{G}{U}_k$ and $\minlsepst{\sat(G,U)}_k$ denote the sets of minimal $s,t$-separators in $\minlsepEst{G}{U}$ and $\minlsepst{\sat(G,U)}$ whose size is exactly $k$, respectively. Since 
	$\minlsepEst{G}{U}= \minlsepst{\sat(G,U)}$, then $\minlsepEst{G}{U}_k= \minlsepst{\sat(G,U)}_k$ for every integer $0 \leq k \leq n$. In particular, this is the case for $k=\kappa_{s,t}(G,\comp{U})=\kappa_{s,t}(\sat(G,U))$. Hence, $\minsep_{s,t}(G,\comp{U})=\minsep_{s,t}(\sat(G,U))$ (see~\eqref{eq:MinimumIncludeExcludeFixed}).
\end{proof}
\subsection{Ranked Enumeration Algorithms for all $s,t$-separators in Ranked Order, and an Enumeration Algorithm for Minimum $s,t$-separators}
\label{sec:enumAlgs}
We directly apply the results of Sections~\ref{sec:includeMinSeps}, and~\ref{sec:excludeU} to the task of enumerating all $s,t$-separators in ranked order in Algorithm~\ref{alg:EnumSep}, and to the task of enumerating the minimum-cardinality $s,t$-separators $\minsep_{s,t}(G)$ in Algorithm~\ref{alg:EnumSepMin}. Both algorithms apply the standard Lawler technique with inclusion and exclusion constraints (see~\cite{DBLP:journals/pvldb/GolenbergKS11} for an overview), leading to simple polynomial-delay algorithms.

\begin{algserieswide}{H}{Algorithm for enumerating all $s,t$-separators in ranked order\label{alg:EnumSep}}	
	\begin{insidealgwide}{RankedEnumSeps}{$G$, $\set{s,t}$}					
		\STATE Let $S$ be a minimum-weight $s,t$-separator of $G$
		\STATE $Q \gets \emptyset$	
		\STATE $Q.\mathrm{push}(\triple{G}{S}{I=\emptyset})$  \COMMENT{$Q$ is sorted by $|S|$}		
		\WHILE{$Q\neq \emptyset$}
		\STATE $\triple{G}{S}{I} \gets Q.\mathrm{pop}()$	
		\STATE \texttt{Print }$S$
		\FORALL{$v_i \in S\setminus I=\set{v_1,\dots,v_q}$}
		\STATE $I_i\gets I\cup \set{v_1,\dots, v_{i-1}}$
		\STATE $H\gets \sat(G,v_i)$ \COMMENT{Exclude $v_i$}
		\STATE $T\gets$ minimum-weight $s,t$-separator in $H\sminus I_i$
		\STATE $Q.\mathrm{push}(\triple{H}{T\cup I_i}{I_i})$
		\ENDFOR	
		\ENDWHILE
	\end{insidealgwide}
\end{algserieswide}
\begin{algserieswide}{H}{Algorithm for enumerating minimum-cardinality $s,t$-separators \label{alg:EnumSepMin}}
	\begin{insidealgwide}{MinimumSeps}{$G$, $\set{s,t}$}					
		\STATE Let $S$ be a minimum-cardinality $s,t$-separator of $G$
		\STATE $\kappa_{s,t}(G)\gets |S|$
		\STATE $Q \gets \emptyset$	
		\STATE $Q.\mathrm{push}(\triple{G}{S}{I=\emptyset})$  
		\WHILE{$Q\neq \emptyset$}
		\STATE $\triple{G}{S}{I} \gets Q.\mathrm{pop}()$	
		\STATE \texttt{Print }$S$
		\FORALL{$v_i \in S\setminus I=\set{v_1,\dots,v_q}$}
		\STATE $I_i\gets I\cup \set{v_1,\dots, v_{i-1}}$
		\STATE $H\gets \sat(G,v_i)$ \COMMENT{Exclude $v_i$ (see Lemma~\ref{lem:MinlSepIU})}
		\STATE $T\gets$ minimum $s,t$-separator of $H\sminus I_i$		
		\IF{$|T|=\kappa_{s,t}(G)-|I_i|$} 
		\STATE $Q.\mathrm{push}(\triple{H}{T}{I_i})$ \COMMENT{See Corollary~\ref{corr:vertexSetInclude}}
		\ENDIF	
		\ENDFOR	
		\ENDWHILE
	\end{insidealgwide}	
\end{algserieswide}

\newpage 
\bibliography{main}
\newpage 
\appendix

\section{Minimal Separators and Chordless $s,t$-paths}
\label{sec:chordlessPath}
In this section we show that given a set $I\subseteq \nodes(G)$, it is NP-hard to decide whether there exists a minimal $s,t$-separator $S\in \minlsepst{G}$ such that $I\subset S$. We prove this by showing a reduction from the problem \textsc{3-in-a-path} that
asks whether there is an induced (or chordless) path containing three given terminals. Bienstock~\cite{BIENSTOCK199185} has shown that deciding whether two terminals belong to an induced cycle is NP-hard. From this, it is easy to show that the \textsc{3-in-a-path}  problem is NP-hard even for graphs whose degree is at most three~\cite{DBLP:journals/dam/DerhyP09}. In fact, even deciding whether there is such a path of length at most $k$ was shown to be $W[1]$-complete with respect to the length parameter $k$~\cite{DBLP:journals/tcs/HaasH06}. The related problem, called \textsc{three-in-a-tree}, for deciding whether there is an induced tree containing three terminals, is in PTIME~\cite{10.1145/3357713.3384235}.

\begin{theorem}
\label{thm:chordlessstpath}
Let $v \in \nodes(G)$. There exists a minimal $s,t$-separator that includes $v$ if and only if there exists a chordless $s,t$-path through $v$.
\eat{
 two paths: $P_{sv}=(s,u_1,\dots,u_k,v)$ from $s$ to $v$, and $P_{vt}=(v,w_1,\dots,w_\ell,t)$ from $v$ to $t$ such that $\set{u_1,\dots,u_k}\cap \set{w_1,\dots,w_\ell}=\emptyset$, and 
for every pair $i\in [1,k]$ and $j\in [1,\ell]$, it holds that $(u_i,w_j)\notin \edges(G)$.
}
\end{theorem}
\begin{proof}
Let $S\in \minlsepst{G}$ where $v\in S$, and let $C_s(G\sminus S)$, $C_t(G\sminus S)$ denote the connected components of $G\sminus S$ that contain $s$ and $t$ respectively. By Lemma~\ref{lem:fullComponents}, there exists a path from $s$ to $v$ where all the internal vertices belong to $C_s(G\sminus S)$. Let $P_{sv}$ denote the shortest such path. Likewise, let $P_{vt}$ denote the shortest path from $v$ to $t$ where all internal vertices belong to $C_t(G\sminus S)$. Clearly, $P_{sv}$ and $P_{vt}$ are both chordless paths. Since $C_s(G\sminus S)\cap C_t(G\sminus S)=\emptyset$, then $\nodes(P_{sv})\cap \nodes(P_{vt})=\set{v}$. Since $S\in \minlsepst{G}$, then there are no edges between vertices in $C_s(G\sminus S)$ and vertices in $C_t(G\sminus S)$. Consequently, there are no edges between vertices in $\nodes(P_{sv})$ and $\nodes(P_{vt})$. Therefore, the path $P_{sv}P_{vt}$ is a chordless $s,t$-path that passes through $v$. In other words, if $v\in S$, then there is an induced $s,t$-path through $v$.

Let $P=s,a_1,\dots,a_k,v,b_1,\dots,b_\ell,t$  denote a simple, chordless $s,t$-path through $v$. If $v \in N_G(s)$ ($v \in N_G(t)$), then $k=0$ ($\ell=0$). Contract all edges on the sub-path $P_a\eqdef (s,a_1,\dots,a_k)$ such that $P_a$ is reduced to an edge $(s,v)$. Likewise, contract all edges on the sub-path $P_b\eqdef (b_1,\dots,b_\ell,t)$ such that $P_b$ is reduced to an edge $(v,t)$. Denote the resulting graph by $G'$. Since $P$ is chordless, then there are no edges between $(a_i,b_j)$ for all $i\in [1,k]$ and all $j\in [1,\ell]$. Therefore, following the contraction, $s$ and $t$ are not adjacent in the new graph $G'$, and hence separable.

Let $S'\in \minlsepst{G'}$ be a minimal $s,t$-separator in $G'$. By construction, $v \in N_{G'}(s)\cap N_{G'}(t)$, and hence $v \in S'$. It is left to show that $S'\in \minlsepst{G}$. Let $C_{s}(S'\sminus G')$ and $C_{t}(S'\sminus G')$ denote the full connected components of $G'\sminus S'$ containing $s$ and $t$ respectively.
Define $D_s(S'\sminus G)\eqdef C_{s}(S'\sminus G')\cup \set{a_1,\dots,a_k}$ and $D_t(S'\sminus G)\eqdef C_{t}(S'\sminus G')\cup \set{b_1,\dots,b_\ell}$. By construction, $D_s(S'\sminus G)$ and $D_t(S'\sminus G)$ are disjoint, and $G[D_s(S'\sminus G)]$ ($G[D_t(S'\sminus G)]$) are both connected components in $G$. Since $C_{s}(S'\sminus G')$ and $C_{t}(S'\sminus G')$ are full components of $S'$ in $G'$, and $D_s(S'\sminus G)\supseteq C_{s}(S'\sminus G')$ and $D_t(S'\sminus G)\supseteq C_{t}(S'\sminus G')$, then $D_s(S'\sminus G)$ and $D_t(S'\sminus G)$ are full components of $S'$ in $G$. By Lemma~\ref{lem:fullComponents}, $S'\in \minlsepst{G}$.
\end{proof}

\eat{
\begin{lemma}
Let $S$ be a minimal $st$-separator in $G$, and let $v\in S$. For every simple, chordless $st$-path through $v$, $P_v=s,a_1,\dots,a_k,v,b_1,\dots,b_\ell$, one of the following holds:
\begin{enumerate}
	\item $a_k \in N_G(v) \cap (C_s(S)\cup S)$ and $b_1 \in N_G(v) \setminus N_G(\set{s,a_1,\dots,a_k})$ 
	\item $a_k,b_1 \in N_G(v) \cap C_t(S)$ and  $a_k \in N_G(v)\setminus N_G(\set{b_1,\dots,b_\ell,t})$
\end{enumerate}
\end{lemma}
\begin{proof}
By Theorem~\ref{thm:chordlessstpath} such a simple, chordless path $P_v=s,a_1,\dots,a_k,v,b_1,\dots,b_\ell$ exists. Clearly, $a_k,b_1\in N_G(v)$ are distinct, and neighbors of $v$.
Suppose, without loss of generality, that $a_k \in N_G(v) \cap (C_s(S)\cup S)$. Assume that $b_1 \notin N_G(v) \cap (C_t(S)\cup S)$. Hence, $b_1 \in N_G(v) \cap C_s(S)$ because $\nodes(G)=C_s(S)\cup S \cup C_t(S)$ forms a partition of $\nodes(G)$. Now, if $b_1 \in N_G(\set{s,a_1,\dots,a_k})$ then the path $P_v=s,a_1,\dots,a_k,v,b_1,\dots,b_\ell$ clearly has a chord because of the edge $(a_i,b_1)$. Therefore, $b_1 \notin N_G(\set{s,a_1,\dots,a_k})$. If $a_k \in N_G(v) \cap C_t(S)$ then it must hold that $b_1 \in C_t(S)$. 
Likewise, if $a_k\in N_G(v) \cap C_t(S)$ and $b_1 \in C_t(S)$ then $a_k \notin N_G(\set{b_1,\dots,b_\ell,t})$.
\end{proof}
}
Theorem~\ref{thm:chordlessstpath} provides a characterization of when a vertex $v$ is included in a minimal $s,t$-separator. By reduction from the \textsc{3-in-a-path} problem we conclude that deciding whether there is a minimal $s,t$-separator containing a subset $I\subseteq \nodes(G)$ is an NP-complete problem.

\eat{
In our case, we do not need to solve this problem directly. Let $S$ be a minimal $st$-separator of $G$, and let $u,v \in S$. By Theorem~\ref{thm:chordlessstpath} there is a simple, chordless $st$-path through $u$ and $v$. Let $a,b \in N_G(u)$. We wish to know if, after the addition of edge $(a,b)$ to $G$, there is still a chordless $st$-path through $v$. To that end we define the following.
}

\end{document}